\documentclass{article}
\usepackage{amsmath,amssymb,amsthm,mathrsfs,color,bm,eucal}
\usepackage{hyperref} \hypersetup{colorlinks=true} 
\usepackage{color}
\newtheorem{theorem}{Theorem}[section]
\newtheorem{corollary}[theorem]{Corollary}

\newtheorem{lemma}[theorem]{Lemma}
\newtheorem{definition}[theorem]{Definition}
\newtheorem{example}[theorem]{Example}

\newtheorem{remark}[theorem]{Remark}

\numberwithin{equation}{section}

\def\Z{{\Bbb Z}} 
\def\R{{\mathscr R}} 
\def\ol{\overline} \def\wh{\widehat}

\begin{document}
\title{Constacyclic Codes over Commutative Finite Principal Ideal Rings}
\author{Yun Fan\\
{\small School of Mathematics and Statistics}\\
{\small Central China Normal University, Wuhan 430079, China}}
\date{}
\maketitle

\insert\footins{\footnotesize{\it Email address}:
yfan@ccnu.edu.cn (Yun Fan).}

\begin{abstract}
For any constacyclic code over a finite commutative chain ring
of length coprime to the characteristic of the ring,
we construct explicitly generator polynomials and check polynomials, and 
exhibit a BCH bound for such constacyclic codes.
As a consequence, such constacyclic codes are principal. 
Further, we get a necessary and sufficient condition that the cyclic codes 
over a finite commutative principal ideal ring are all principal. 
This condition is still sufficient
for constacyclic codes over such rings being principal.

\end{abstract}

\section{Introduction}
Any ring in this paper is a commutative ring with identity.
An ideal of a ring is said to be {\em principal} if it is generated by one element.
A ring is called a {\em principal ideal ring}, abbreviated by PIR, 
if its any ideal is principal. 

For a finite ring~$R$,  $\R=R[X]/\langle X^n-\lambda\rangle$ denotes the
quotient ring of the polynomial ring $R[X]$ over the ideal 
$\langle X^n-\lambda\rangle$ generated by $X^n-\lambda$, where
$\lambda$ is a unit (invertible element) of~$R$. 
Any ideal~$C$ of $\R$ is called a {\em $\lambda$-constacyclic code} 
of length $n$ over $R$;
and~$C$ is said to be {\em cyclic} once $\lambda=1$.
If the ideal $C$ is principal (i.e.,  generated by one element), 
we also say that the $\lambda$-constacyclic code $C$ is {\em principal}. 
Any $a(X)=\sum_{i=0}^{n-1}a_iX^i\in\R$ is 
a word $(a_0,a_1,\cdots,a_{n\!-1\!})\in R^n$; 
the {\em Hamming weight} ${\rm w_H}(a(X))$ is defined to be 
the number of the non-zero coefficients $a_i\ne 0$. 
And the {\em minimum Hamming weight} ${\rm w_H}(C)$ of the code $C$
is defined to be the minimal Hamming weight of non-zero
{\em code words} in $C$.

A ring is called a {\em finite chain ring} if it is finite and
its ideals form a chain with respect to the inclusion relation.
Finite fields are a special kind of finite chain rings.
A finite ring is a principal ideal ring if and only if it is a direct sum of 
finite chain rings (cf. \cite{McD} or Eq.\eqref{being PIR} below). 
Thus the study on constacyclic codes over finite principal ideal rings
would be reduced to the case over finite chain rings
(cf. Lemma \ref{R=R_1+...} below).

Cyclic and constacyclic codes over finite fields take an 
important part in coding theory and practice, 
since they have nice algebraic structures, 
e.g., they are invariant by (consta-)cyclic shifts; and they
have nice distance behaviour as well, 
e.g., BCH bounds for minimum Hamming weights. 
For example, see \cite{HP}.

Initiated by \cite{HKCSS}, coding over finite rings was developed very much. 
Especially, cyclic and constacyclic codes over finite chain rings 
attracted a lot of attention, e.g., \cite{FNKS, HL, NS}.
If the code length $n$ is coprime to the {\em characteristic} of 
the finite chain ring $R$, \cite{DL} showed that
the cyclic and negacyclic codes over~$R$ are principal; 
and a kind of generator polynomials for such codes
are exhibited in \cite[Theorem 3.6]{DL}, 
and in \cite[Theorem 2.4]{BFT16} (based on \cite{DL}) also.
As for constacyclic codes, 
there are a lot of works varied from case to case: 
over particular finite chain rings~$R$, for particular $\lambda\in R$
and particular code length~$n$; e.g., 
\cite{C13, CCDFM, D10, DFLLS, DKK, LL, LM, LS}.

We are concerned with two questions: when are the constacyclic codes over 
a finite principal ideal ring all principal? and, if it is the case, what about the
generator polynomials, check polynomials and BCH bounds?

In Section 2 we sketch preliminaries about finite chain rings 
and finite principal ideal rings. 

In Section 3, we consider the constacyclic codes over a finite chain 
ring~$R$ of length $n$ coprime to the characteristic of $R$.
We show in a ring-theoretic way that 
the quotient ring $R[X]/\langle X^n=\lambda\rangle$ is principal.
Inspiring by \cite[Theorem~3.6]{DL} on cyclic codes, 
we transform the generator (of ring-theoretic style) of any constacyclic code $C$  
in to a generator polynomial of a typical coding-theoretic form, 
and get a check polynomial of~$C$ as well. 
The generator polynomial and the check polynomial
of $C$ are proved unique up to some sense. 
Moreover, we exhibit a BCH bound for such constacyclic codes. 

In Section 4, we get a necessary and sufficient condition that
the {\em isometrically cyclic} codes (see Definition \ref{d iso cyclic} below) over 
a finite principal ideal ring are all principal. 
That condition is still sufficient for all constacyclic codes 
over finite principal ideal rings being principal.
However,  an example shows that it is no longer necessary 
for such codes being principal. 
    
\section{Preliminaries}\label{preliminaries}
In this paper any ring $R$ is commutative and with identity $1_R$ (or $1$ for short). 
Subrings and ring homomorphisms are identity-preserving.  
By $R^\times$ we denote the multiplicative group consisting of
 all units (invertible elements) of $R$. 

For $a\in R$, the ideal generated by $a$ is denoted by $\langle a\rangle$,
called a {\em principal ideal} of $R$.
Then $\langle a\rangle=Ra:=\{ba\mid b\in R\}$. 
So both $Ra$ and $\langle a\rangle$ denote the principal ideal generated by $a$.
Two elements $a, b$ of $R$ are said to be {\em coprime} if 
${R}a+{R}b={R}$
(equivalently, there are elements $a',b'$ such that $aa'+bb'=1$).
We denote $a|b$ if $b\in {R}a$ 
(equivalently, there is an element $a'$ such that $b=a'a$).
We list some known facts in the following lemma.

\begin{lemma}\label{coprime}
Let $a, b\in R$ be coprime, and $c\in R$. The following hold:

{\bf(1)}
  If $a$ and $c$ are also coprime, then $a$ and $bc$ are coprime.

{\bf(2)} ${R}(ac)+{R}(bc)={R}c$.

{\bf(3)}  ${R}(ac)\cap{R}(bc)\!=\!{R}(abc)$. 
(Equivalently, if $ac | d$ and $bc | d$, then $abc | d$.) 

{\bf(4)}
 If $ab=0$, then ${R}={R}a\oplus{R}b$,
${R}a$ is a ring (but not a subring of $R$ in general), 
$a$ is a unit of ${R}a$ (not a unit of $R$ in general), and $Ra^2=Ra$.
\end{lemma}

\begin{proof} (1) and (2) are checked as usual.

(3). Assume that $d=ace=bcf$ and $aa'+bb'=1$; then 
$d=daa'+dbb'=bcfaa'+acebb'=abc(fa'+eb')$. 

(4). If $ab=0$, from (3) (by taking $c=1$), we get ${R}a\cap{R}b={R}(ab)=0$;
so ${R}={R}a\oplus{R}b$. Let $1=\iota_a+\iota_b$ with
$\iota_a\in{R}a$ and $\iota_b\in{R}b$. Then $\iota_a$ is the identity of ${R}a$, 
hence $Ra$ is a ring.
There is a $d\in{R}$ such that $\iota_a=da=(\iota_ad)a$, 
i.e., $a\in{R}a$ is invertible. 
The principal ideal generated by a unit of a ring is the ring itself. 
Thus $Ra^2=(Ra)a=Ra$.
\end{proof}

\begin{remark}\label{direct sum ideal}\rm
An ideal $D$ of $R$ is a ring if and only if $D$ has an identity $\iota$
(in particular, $\iota$ is an {\em idempotent}, i.e., $\iota^2=\iota$),
if and only if $R=D\oplus D'$ with $D=R\iota$ and $D'=R(1-\iota)$. 
In that case, any ideal $C$ of $R$ is decomposed as
$C=(C\cap D)\oplus (C\cap D')$ with 
$C\cap D=\iota C$ and $C\cap D'=(1-\iota)C$.
\end{remark}
 
Let $\Z$ be the integer ring, and $\Z_m$ be the residue ring modulo $m$.
For any finite ring $R$,
the image of the homomorphism $\Z\to R$, $k\mapsto k1_{R}$, 
 is the smallest subring of $R$ generated by $1_R$, 
 which must be isomorphic to a residue ring $\Z_{m}$.
The smallest subring of $R$ is called the {\em prime subring} of $R$, 
and denoted by $\Z_{m}$ again. The integer $m$ is called the 
{\em characteristic} of $R$, denoted by ${\rm char} R$.

By fundamentals of ring theory (see \cite{J} for example), 
the {\em Jacobson radical} $J(R)$ of a finite ring $R$ is 
the unique maximal nilpotent ideal of $R$, 
and $\ol R:=R/J(R)$ is a direct sum of finite fields.  
And, any idempotent $\ol\iota$ of $\ol R$ is uniquely lifted to 
an idempotent $\iota$ of $R$.  

A finite ring $R$ is said to be {\em local} if $R/J(R)$ is a field 
(called the {\em residue field}). Note that
the sum of a unit and a nilpotent element is a unit. 
A finite ring $R$ is local if and only if the difference set $R-J(R)=R^\times$.
By the idempotent lifting mentioned above, 
$R$ is local if and only if the identity $1_R$ is the unique non-zero idempotent. 
By idempotent lifting again, any finite ring $R$ can be written as:
\begin{equation}\label{+local}
R=R_1\oplus\cdots\oplus R_m\quad \mbox{with all $R_i$'s being local;}
\end{equation}  
and it is easy to see that 
${\rm char} R={\rm lcm}\big({\rm char} R_1,\cdots,{\rm char} R_m\big)$,
 the least common multiple of ${\rm char} R_1,\cdots,{\rm char} R_m$.
Some well-known preliminaries about finite local rings
 are listed in the following lemma, please cf. \cite{GM,McD,N08} for details. 

\begin{lemma}\label{facts on local} 
Let $R$ be a finite local ring, 
$F:=\ol R=R/J(R)$ be the residue field.
 
{\rm(1)} The characteristic ${\rm char}\,R=p^{s}$ for
a prime $p$ and $s\ge 1$. Hence ${{\rm char}F=p}$, which we call
the {\em residue characteristic} of the local ring $R$. 

{\rm(2) (Hensel's Lemma)} Let $f(X)\in R[X]$ be monic. If the residue polynomial 
$\ol f(X)=\zeta(X)\cdot\eta(X)$ where $\zeta(X),\eta(X)\in F[X]$
are monic and coprime, then there are unique 
monic $g(X),h(X)\in R[X]$ which are coprime 
such that their residues $\ol g(X)=\zeta(X)$, $\ol h(X)=\eta(X)$, 
and $f(X)=g(X)h(X)$. 

{\rm(3) (Unique factorization)}
If $f(\!X\!)\!\in\! R[X]$ is monic and its residue  
$\ol f(\!X\!)\!\in\! F[X]$ has no repeated roots, 
then $f(\!X\!)$ is uniquely (up to reordering) factored into a product 
of {\em basic irreducible} polynomials   
{ (A monic $\varphi(X)\in R[X]$ is said to be {\em basic irreducible}
if its residue $\ol\varphi(X)\in F[X]$ is irreducible)}.
\end{lemma}

Finite chain rings and finite principal ideal rings have been introduced before. 
The following lemma implies that
finite chain rings are just finite local principal ideal rings. 

\begin{lemma}\label{being chain}
A finite ring $R$ is a chain ring if and only if $R$ has a nilpotent element $\pi$ 
such that $R/R\pi$ is a field. 
\end{lemma}

\begin{proof} Let $R$ be a chain ring. Then $R$ has a unique maximal ideal
which must be the Jacobson radical $J(R)$. Since $J(R)$ is maximal, $R/J(R)$ is a field.
Suppose that $J(R)\ne\langle a\rangle$ for any $a\in J(R)$, then
there are $a,b\in J(R)$ such that $\langle a\rangle$ and~$\langle b\rangle$
cannot contain each other, that is a contradiction to that $R$ is a chain ring.
Thus there is a $\pi\in J(R)$ such that $J(R)=R\pi$.

Conversely, assume that $\pi\in R$ such that $\pi^\ell=0$,  $\pi^{\ell-1}\ne 0$ 
and $R/R\pi$ is a field. Then $R\pi$ is a nilpotent ideal hence $R\pi=J(R)$, and
\begin{equation}\label{chain of ideals}
 R\supsetneq R\pi\supsetneq\cdots
  \supsetneq R\pi^{\ell-1}\supsetneq R\pi^\ell=0
\end{equation}
is a chain of ideals of $R$.  
For any non-zero ideal $C$ of $R$,  
there is an index $j$, $0\le j<\ell$, such that
$C\subseteq R\pi^j$ but $C\,{\not\subseteq}\,R\pi^{j+1}$.
Then there is a $c\in C\!-\!R\pi^{j+1}$, 
 $c=u\pi^j$ with $u\in R\!-\!R\pi$ (i.e., $u$ is a unit).
So $C\supseteq Rc=Ru\pi^j=R\pi^j$. That is, $C=R\pi^j$.
Thus, the chain \eqref{chain of ideals} are all ideals of $R$,
and $R$ is a chain ring.
\end{proof}

Let $R$ be a finite chain ring as in Lemma \ref{being chain}.
By the lemma and the proof, the radical $J(R)=R\pi$ and 
all ideals of $R$ are listed in Eq.\eqref{chain of ideals},
where $\ell$ is called the {\em nilpotency index} of $J(R)$ (of $\pi$).
We further assume that ${\rm char}\,R=p^{s}$ for a prime $p$ and $s\ge 1$
(hence $\Z_{p^s}$ is the prime subring of $R$), 
and that the residue field $F:=\ol R=R/R\pi={\rm GF}(p^r)$ 
is the finite field (Galois field) of cardinality~$p^r$.
Then the unit group $F^\times\!=\!\langle\ol\gamma\rangle$ generated 
by an element $\ol\gamma$ of order $p^r\!-\!1$. So 
$X^{p^r-1}\!-\!1=\ol g(X)\ol h(X)$, where $\ol g(X),\ol h(X)\in\Z_p[X]$, 
 $\ol g(X)$ is monic irreducible and 
 $\ol g(\ol\gamma)\!=\!0$, hence $\deg\ol g(X)\!=\!r$. 
 By Hensel's Lemma (Lemma \ref{facts on local}(2)), we have 
 $g(X), h(X)\in \Z_{p^s}[X]$ and $\gamma\in R$ such that 
 $X^{p^r-1}\!-\!1=g(X)h(X)$, $g(X)$ is monic basic irreducible, 
$g(\gamma)\!=\!0$  and $\deg g(X)\!=\!r$. 
In $R$, we get a subring 
$\Z_{p^s}[\gamma]$ generated by $\gamma$ over $\Z_{p^s}$, and
$\Z_{p^s}[\gamma]\cong\Z_p[X]/\langle g(X)\rangle$,  which is
called the {\em Galois ring} of invariants $p^s, r$ 
and denoted by ${\rm GR}(p^s, r)$ (cf. \cite{W03}).
Then there is an integer $k\ge 1$ such that
\begin{equation}\label{Eisenstein equation}
  p\in R\pi^k-R\pi^{k+1}, \quad \mbox{equvalently, ~
  $\pi^k=pu$~ for a $u\in R^\times$;}
\end{equation}
hence $(s-1)k<\ell\le sk$.

\begin{remark}\rm
As described above, 
a finite chain ring $R$ has five invariants $p,r,s,k,\ell$
satisfying that $(s-1)k<\ell\le sk$,
where $p,r,s$ are the invariants of the Galois subring 
$\Z_{p^s}[\gamma]={\rm GR}(p^s,r)$. 
In the following two cases finite chain rings are uniquely determined 
up to isomorphism by their invariants $p,r,s,k,\ell$:

$\bullet$~ 
$k=1$ (hence $\ell=s$, three invariants $p,r,s$ remain), 
i.e., $R={\rm GR}(p^s,r)$;

$\bullet$~ 
$s=1$ (hence $\ell=k$, three invariants $p,r,\ell$ remain), 
i.e., $R=F[X]/\langle X^\ell\rangle$.

\noindent
In general, however, there may be many finite chain rings 
non-isomorphic each other having the same invariants $p,r,s,k,\ell$.
As far as we know, it is still an open question how to classify finite chain rings. 
We refer to \cite{CL, Hou01} for details.
\end{remark}

\begin{lemma}\label{Galois extension}
Let $R$ be a finite chain ring with invariants $p,r,s,k,\ell$. 
Let $\varphi(X)\in R[X]$ be basic irreducible 
(i.e., $\varphi(X)$ is monic and $\ol\varphi(X)\in F[X]$ is irreducible) 
with $d:=\deg\varphi(X)$. 
Then $R[X]\big/\langle\varphi(X)\rangle$ is a finite chain ring 
with invariants $p,rd,s,k,\ell$. 
\end{lemma} 

\begin{proof} Denote $R_\varphi=R[X]\big/\langle\varphi(X)\rangle$.
Obviously, $\pi\in R_{\varphi}$, $\pi^\ell=0$ but $\pi^{\ell-1}\ne 0$. And
$$R_{\varphi}/R_{\varphi}\pi \cong \ol R[X]\big/\langle \ol\varphi(X)\rangle
  =F[X]/\langle \ol\varphi(X)\rangle={\rm GF}(p^{rd}),$$
which is a field extension over $F={\rm GF}(p^r)$ of extension degree $d$.
By Lemma~\ref{being chain}, 
$R_{\varphi}$ is a chain ring with radical $J(R_{\varphi})=R_{\varphi}\pi$.
And, if $\pi$ in $R$ satisfies Eq.\eqref{Eisenstein equation}, 
then so is the $\pi$ in $R_{\varphi}$ 
(because $R^\times\subseteq R_\varphi^\times$).
\end{proof}

By Eq.\eqref{+local} and Lemma \ref{being chain}, 
any finite principal ideal ring $R$ is written as
\begin{equation}\label{being PIR}
 R=R_1\oplus\cdots\oplus R_m\quad 
 \mbox{with all $R_i$'s being finite chain rings.} 
\end{equation}
Conversely, the ring $R$ in Eq.\eqref{being PIR} must
be a finite principal ideal ring; because: assuming that 
$J(R_i)=R_i\pi_i$ of nilpotency index $\ell_i$,
for any ideal~$C$ of~$R$ we have, by Remark \ref{direct sum ideal} 
and Eq.\eqref{chain of ideals}, that
\begin{equation*}
 C=C\cap R_1\oplus\cdots\oplus C\cap R_m
 =R_1\pi_1^{e_1}\oplus\cdots\oplus R_m\pi_m^{e_m},
 ~ 0\le e_i\!\le\! \ell_i,~ i\!=\!1,\cdots,m;
\end{equation*}
hence $C$ is generated by one element $ \pi_1^{e_1}+\cdots+\pi_m^{e_m}$.
In this way, the ideals of the ring $R$ in Eq.\eqref{being PIR}
are one-to-one corresponding to the integer sequences 
$(e_1,\cdots,e_m)$,  $0\le e_i\le \ell_i$ for $i=1,\cdots,m$.
For such a sequence $e_1,\cdots,e_m$
and any $u_i\in R_i^\times$, 
the element $\pi_1^{e_1}u_1+\cdots+\pi_m^{e_m}u_m$ 
also generates the ideal:
\begin{equation}\label{PIR}
 R_1\pi_1^{e_1}\oplus\cdots\oplus R_m\pi_m^{e_m}
 =R\cdot(\pi_1^{e_1}u_1+\cdots+\pi_m^{e_m}u_m).
\end{equation}
Please see \cite{McD} or \cite{N08}.

\section{Constacyclic codes over finite chain rings}\label{cyclic codes}

Through out this section, we assume that
\begin{itemize}
\item
$R$ is a finite chain ring with radical $J(\!R)\!=\!R\pi$ of nilpotency index~$\ell$,
where \linebreak 
$\pi\!\in\! J(\!R)\!-\!J(\!R)^2$, 
and the residue field $F\!=\!\ol R\!=\!R/\!J(\!R)={\rm GF}(q)$ with~$q\!=\!p^r$;
\item $n>1$ is an integer, $\gcd(n,p)=1$;  
\item
$\R=R[X]\big/\langle X^n-\lambda\rangle$,  where 
$\lambda\in R^\times$ (hence $\ol\lambda\in F^\times$). 
\end{itemize}

With the assumption above, in this section
we consider the $\lambda$-constacyclic codes of length $n$ over $R$,
i.e., the ideals of $\R$.
We will characterize the algebraic structures of these codes; 
and exhibit a BCH bound for them;  we will conclude this section by
an example to illustrate the obtained results.   

We begin by fixing notation.
In $F[X]$ we have an irreducible decomposition 
(in Remark \ref{q-cosets} we'll show how to get the decomposition):
\begin{equation}\label{X^n-lambda in F}
X^n-\ol\lambda=\ol\varphi_1(X)\cdots\ol\varphi_m(X),
\qquad\deg\ol\varphi_i(X)=d_i, ~~ i=1,\cdots,m,
\end{equation}
where
$\ol\varphi_1(X)$, $\cdots$, $\ol\varphi_m(X)\in F[X]$ 
are monic irreducible and coprime to each other.
By Hensel's Lemma, 
$\ol\varphi_i(X)$, $i=1,\cdots,m$, are uniquely lifted to monic polynomials 
$\varphi_1(X)$, $\cdots$, $\varphi_m(X)\in R[X]$ 
such that, in $R[X]$, $\varphi_1(X)$, $\cdots$, $\varphi_m(X)$ are basic irreducible 
and coprime to each other, and
\begin{equation}\label{X^n-lambda in R}
X^n-\lambda=\varphi_1(X)\cdots\varphi_m(X).
\end{equation}
With respect to this decomposition we set 
\begin{equation}\label{def I}I=\{1,2,\cdots,m\}\end{equation} 
to be the subscript set. For any subset $I'\subset I$, 
 define
\begin{equation}\label{phi_I}\textstyle
 \varphi_{I'}(X)=\prod_{i\in I'}\varphi_i(X),\quad 
 \wh\varphi_{I'}(X)=\varphi_{I- I'}(X)=(X^n-\lambda)/\varphi_{I'}(X);
\end{equation}
where we appoint that $\varphi_{\emptyset}(X)=1$, 
hence $\wh\varphi_{\emptyset}(X)=X^n-\lambda$.
If $I'=\{i\}$ contains exactly one subscript $i$, we abbreviate them by
\begin{equation}\label{phi_i's}
\varphi_i(X)=\varphi_{\{i\}}(X),\quad 
\wh\varphi_{i}(X)=\wh\varphi_{\{i\}}(X)=(X^n-\lambda)/\varphi_{i}(X).
\end{equation}

\subsection{Algebraic structures of constacyclic codes over $R$}

\begin{theorem}\label{constacyclic} 
Let notation be as in Eq.\eqref{X^n-lambda in F}--\eqref{phi_i's}. 
Then $\R $ is a direct sum of ideals
$\R \wh\varphi_i(X)$ for $i\in I$ as follows:
\begin{equation}\label{inner CRT}
  \R =\R \wh\varphi_1(X)\oplus \R \wh\varphi_2(X)
   \oplus\cdots\oplus \R \wh\varphi_m(X);
\end{equation}
as rings, the ideal
$\R \wh\varphi_i(X) \cong  R[X]/\langle\varphi_i(X)\rangle$
 is a finite chain ring such that:
 
{\bf(1)} the radical 
 $J\big(\R \wh\varphi_i(X)\big)=\R \pi\wh\varphi_i(X)$ 
 of nilpotency index~$\ell$; 
 
{\bf(2)} the residue field 
$\R\wh\varphi_i(X)/\R \pi\wh\varphi_i(X)={\rm GF}(p^{rd_i})$;
 
{\bf(3)} $\wh\varphi_i(X)$ is a unit of the ring $\R \wh\varphi_i(X)$.
 
\noindent
 In particular, $\R$ is a finite principal ideal ring.
\end{theorem}

\begin{proof}
By Chinese Remainder Theorem, the natural homomorphism 
\begin{equation}\label{nat homo}
\begin{array}{ccccc}
 R[X]&\longrightarrow& R[X]\big/R[X]\varphi_1(X)
   ~~\oplus~\cdots~\oplus~~ R[X]\big/R[X]\varphi_m(X)\\[1pt]
   f(X)&\longmapsto &\big(\, f(X)~({\rm mod}\;\varphi_1(X))~, 
   ~~~ \cdots ~,~~~ f(X)~({\rm mod}\;\varphi_m(X))\,\big) 
\end{array}\end{equation}
induces an isomorphism
\begin{equation}\label{CRT}
\R \cong 
R[X]\big/R[X]\varphi_1(X)\oplus\cdots\oplus R[X]\big/R[X]\varphi_m(X).
\end{equation}
Obviously, $\varphi_i(X)\wh\varphi_i(X)\equiv 0\!\pmod{X^n-\lambda}$.
By Lemma \ref{coprime}(1), $\varphi_i(X)$ and $\wh\varphi_i(X)$ are coprime; 
and by Lemma \ref{coprime}(4),
$$\R =\R \wh\varphi_i(X)\oplus \R \varphi_i(X).$$
And  $\R \varphi_i(X)$ is the kernel of the natural
homomorphism 
$$\R~\rightarrow ~R[X]\big/R[X]\varphi_i(X), ~~ 
 f(X)~({\rm mod}\,X^n\!-\!\lambda)~ 
  \mapsto~ f(X)~({\rm mod}\,\varphi_i(X)).$$
Thus the natural homomorphism \eqref{nat homo} induces an isomorphism of rings
\begin{equation}\label{inn-out}
\begin{array}{ccc}
 \R \wh\varphi_i(X) &\cong & R[X]/\langle\varphi_i(X)\rangle,\\[1pt]
  f(X)\wh\varphi_i(X)~({\rm mod}~X^n-\lambda)&\longmapsto&
  f(X)\wh\varphi_i(X)~({\rm mod}~\varphi_i(X)).
\end{array}
\end{equation} 
Then the isomorphism \eqref{CRT} corresponds to 
the inner direct decomposition~\eqref{inner CRT}.

(1) and (2).~ By Eq.\eqref{inn-out} and Lemma \ref{Galois extension}, 
$\R \wh\varphi_i(X)$ is a finite chain ring
with radical $J(\R \wh\varphi_i(X))=\R \pi\wh\varphi_i(X)$
of nilpotency index $\ell$, and the residue field 
$$\R \wh\varphi_i(X)\big/\R \pi\wh\varphi_i(X)
   \cong F[X]\big/\langle\ol\varphi_i(X) \rangle={\rm GF}(p^{rd_i}).$$

(3).~ Let $\iota_i(X)$ be the identity of the ring $\R \wh\varphi_i(X)$
(in particular, $\iota_i(X)$ is an idempotent, see Remark \ref{direct sum ideal}).
Note that 
$\wh\varphi_i(X)=1_\R \wh\varphi_i(X)\in \R \wh\varphi_i(X)$.
Since $\iota_i(X)\in \R \wh\varphi_i(X)$,
there is a $\psi_i(X)\in\R $ such that
$\iota_i(X)=\psi_i(X)\cdot\wh\varphi_i(X)$. Thus
\begin{equation}\label{idempotent}
\iota_i(X)=\iota_i(X)^2=\iota_i(X)\psi_i(X)\cdot\wh\varphi_i(X).
\end{equation}
That is, $\wh\varphi_i(X)$ is invertible in the ring $\R \wh\varphi_i(X)$
since $\iota_i(X)\psi_i(X)\in \R \wh\varphi_i(X)$.
\end{proof}

\begin{remark}\rm
By Eq.\eqref{inn-out}, it is easy to get the identity
$\iota_i(X)$ of the ring $\R \wh\varphi_i(X)$ as follows:
$\psi_i(X)\wh\varphi_i(X)+\chi_i(X)\varphi_i(X)=1$
for some $\psi_i(X),\chi_i(X)\in R[X]$
(since $\varphi_i(X)$ and $\wh\varphi_i(X)$ are coprime);
then
$\iota_i(X):=\psi_i(X)\wh\varphi_i(X)\in\R \wh\varphi_i(X)$,
which is mapped by Eq.\eqref{inn-out} to 
the identity of $R[X]/\langle\varphi_i(X)\rangle$.
\end{remark}

\begin{corollary}\label{constacyclic cor1}
Keep the notation in Theorem \ref{constacyclic}. All ideals $C$ of $\R $ 
are one to one corresponding to the integer sequences $(e_1,\cdots,e_m)$,
$0\le e_i\le\ell$ for $i\in I$, such that
\begin{equation}\label{ideals of R}
 C=\R \pi^{e_1}\wh\varphi_1(X)\oplus\cdots\oplus\R \pi^{e_m}\wh\varphi_m(X)
 =\big\langle\pi^{e_1}\wh\varphi_1(X)+\cdots+\pi^{e_m}\wh\varphi_m(X)\big\rangle,
\end{equation}
with $|C|=p^{r(m\ell-d_1e_1-\cdots-d_me_m)}$; 
and the total number of ideals of $\R $ is $(\ell+1)^m$.
\end{corollary}

\begin{proof}
By Theorem \ref{constacyclic} and Eq.\eqref{PIR}, any ideal $C$ of $\R $ 
is characterized by a integer sequence $(e_1,\cdots,e_m)$,
$0\le e_i\le\ell$ for $i=1,\cdots,m$,  as in \eqref{ideals of R}. 
For each $e_i$, there are $\ell+1$ choices. We get $(\ell+1)^m$ ideals.
\end{proof}

Following Eq.\eqref{PIR} which is classical in ring theory,
Eq.\eqref{ideals of R} shows a kind of generators of any ideal~$C$ of~$\R$.
Inspired by the generators of cyclic codes obtained in~\cite{DL},  
we show a kind of generators of the $\lambda$-constacyclic codes  
which is typical in coding theory. 

\begin{theorem}\label{generator of ideal}
Let notation be as above, and $C$ be an ideal of $\R$ as in Eq.\eqref{ideals of R}.
For $v=0,1,\cdots,\ell-1$, define 
$I_v\subseteq I$ by $I_v=\{i\mid i\in I,\; e_i\le v\}$. 
Then we get a sequence of subsets 
$I_0\subseteq I_1\subseteq\cdots\subseteq I_{\ell-1}\subseteq I$,
and the corresponding sequence of monic polynomials
$\wh\varphi_{I_0}(X)$, $\wh\varphi_{I_1}(X)$, 
$\cdots$, $\wh\varphi_{I_{\ell-1}}(X)$
satisfy the following three:

{\bf(1)}~
$\wh\varphi_{I_0}(X)\,\big|\,(X^n-\lambda)$, and
$\wh\varphi_{I_v}(X)\,\big|\,\wh\varphi_{I_{v-1}}(X)$ for $v=1,\cdots, \ell-1$;

{\bf(2)}~
$C\cap\R \pi^v=\big\langle \pi^v\wh\varphi_{I_v}(X)
 +\cdots+\pi^{\ell-1}\wh\varphi_{I_{\ell-1}}(X)  \big\rangle$,
 $v=0,1,\cdots, \ell-1$.

{\bf(3)}~ 
$ \wh\varphi_{I_0}(X)+\pi\wh\varphi_{I_1}(X)
 +\cdots+\pi^{\ell-1}\wh\varphi_{I_{\ell-1}}(X)$
 is a generator of the code $C$. 

\end{theorem}

\begin{proof} 
About the notation \eqref{phi_I},
for $I',I''\subseteq I$, it is obvious that:
\begin{align}\label{phi division}
I'\subseteq I'' ~\iff~ \varphi_{I'}(X)\,\big|\,\varphi_{I''}(X)
      ~ \iff~ \wh\varphi_{I''}(X)\,\big|\,\wh\varphi_{I'}(X);
\end{align}
and by  Lemma~\ref{coprime}(2) it is easy to check that,
\begin{align}\label{phi cdot 3}
\R \wh\varphi_{I'}(X)+\R \wh\varphi_{I''}(X)
 =\R \wh\varphi_{I'\cup I''}(X). 
\end{align}

Obviously, $I_0\subseteq I_1\subseteq\cdots\subseteq I_{\ell-1}\subseteq I$.
By Eq.\eqref{phi division}, (1) holds. 

For (2), we have an observation as follows:
\begin{equation}\label{observe}
\R \pi^{e_i}\wh\varphi_i(X)\cap\R \pi^{v}
=\begin{cases}
\R \pi^{v}\wh\varphi_i(X), & e_i\le v;\\
\R \pi^{e_i}\wh\varphi_i(X), & e_i> v.
\end{cases}
\end{equation}
For $v=\ell-1$,
$$\textstyle
C\cap\R \pi^{\ell-1}=
\bigoplus_{i\in I_{\ell-1}}\R \pi^{\ell-1}\wh\varphi_i(X)
=\pi^{\ell-1}\big(\bigoplus_{i\in I_{\ell-1}}\R \wh\varphi_i(X)\big).
$$
By Eq.\eqref{phi cdot 3}, 
$\bigoplus_{i\in I_{\ell-1}}\R \wh\varphi_i(X)
 =\R \wh\varphi_{I_{\ell-1}}(X)$. So
$$\textstyle
C\cap\R \pi^{\ell-1}=
 \pi^{\ell-1}\R \wh\varphi_{I_{\ell-1}}(X)=
 \R \pi^{\ell-1}\wh\varphi_{I_{\ell-1}}(X).
$$
Thus, (2) holds for $v=\ell-1$. 

Next assume that $0\le v<\ell-1$ and
(2) holds for $v+1$, i.e., 
$$C\cap\R \pi^{v+1}
 =\big\langle \pi^{v+1}\wh\varphi_{I_{v+1}}(X)+\cdots+
   \pi^{\ell-1}\wh\varphi_{I_{\ell-1}}(X)\big\rangle
   =\big\langle  \pi^{v+1} f_v(X)\big\rangle,
$$
where $f_v(X)=\wh\varphi_{I_{v+1}}(X)+\cdots+
   \pi^{\ell-1-v-1}\wh\varphi_{I_{\ell-1}}(X)$.
By Eq.\eqref{observe},
$$\textstyle
C\cap\R \pi^v=C\cap\R \pi^{v+1}+
\big(\bigoplus_{i\in I_v}\R \pi^v\wh\varphi_i(X)\big)
=C\cap\R \pi^{v+1} 
 + \pi^v\big(\bigoplus_{i\in I_v}\R \wh\varphi_i(X)\big)
$$
By Eq.\eqref{phi cdot 3} again, 
$\pi^v\big(\bigoplus_{i\in I_v}\wh\varphi_i(X)\big)
=\pi^v\R \wh\varphi_{I_v}(X)$.
We obtain
\begin{equation}\label{C cap pi^v}
C\cap\R \pi^v=
\big\langle\pi^v\wh\varphi_{I_v}(X),\;\pi^{v+1}f_v(X)\big\rangle.
\end{equation}
Noting that $\pi^\ell=0$, we have
\begin{align*}
&\big(\wh\varphi_{I_v}\!(X\!)+\pi f_v(X) \big)
\left(\wh\varphi_{I_v}\!(X\!)^{\ell\!-\!1}
 -\wh\varphi_{I_v}\!(X\!)^{\ell\!-\!2}\!\cdot\!\pi f_v(X)+
 \cdots+(\!-\!1)^{\ell\!-\!1}\big(\pi f_v(X)\big)^{\ell\!-\!1}\right)\\
&=\wh\varphi_{I_v}(X)^{\ell} + (-1)^{\ell-1}\!\big(\pi f_v(X)\big)^\ell
 =\wh\varphi_{I_v}(X)^{\ell},
\end{align*}
which implies $\wh\varphi_{I_v}(X)^{\ell}
\in\big\langle \wh\varphi_{I_v}(X)+\pi f_v(X) \big\rangle$.
By Lemma~\ref{coprime}(2) and (4),  
$\wh\varphi_{I_v}(X)$ and $\varphi_{I_v}(X)$ are coprime, 
$\R =\R \wh\varphi_{I_v}(X)\oplus
\R \varphi_{I_v}(X)$,
and $\R \wh\varphi_{I_v}(X)^{\ell}=\R \wh\varphi_{I_v}(X)$. Thus
$$
\wh\varphi_{I_v}(X)\in \R \wh\varphi_{I_v}(X)^{\ell}
\subseteq\big\langle \wh\varphi_{I_v}(X)+\pi f_v(X) \big\rangle.
$$
Hence $\pi^v\wh\varphi_{I_v}(X)
\in\big\langle \pi^v\wh\varphi_{I_v}(X)+\pi^{v+1} f_v(X) \big\rangle$. Consequently,
$$
\big\langle \pi^v\wh\varphi_{I_v}(X),\,\pi^{v+1} f_v(X)\big\rangle=
\big\langle \pi^v\wh\varphi_{I_v}(X)+\pi^{v+1} f_v(X) \big\rangle.
$$
Combining it with Eq.\eqref{C cap pi^v}, we see that (2) hold.

Finally, (3) follows from (2) by taking $v=0$.
\end{proof}

In the proof, we constructed the generators by ``going upstairs''.
The first step ($v=\ell-1$) will play a crucial part in the next subsection.
In the cyclic case, i.e., $\lambda=1$, 
the last step ($v=0$) is just \cite[Theorem 3.6]{DL}. 

Next we show a uniqueness up to some sense of
the generator polynomial and check polynomial of $C$. 
Recall from literature (e.g.,~\cite{JLLX}) that,
if there is an $h(X)\in\R$ such that 
$$
 f(X)\in C~\iff~ f(X)h(X)\equiv 0\!\!\pmod{X^n-\lambda},
\qquad \forall\; f(X)\in\R,
$$
then $C$ is said to be {\em checkable} and 
$h(X)$ is called a {\em check polynomial} of $C$.

\begin{theorem}\label{constacyclic cor2}
Let $C$ be a $\lambda$-constacyclic code of length $n$ over $R$.

{\bf(1)}{\rm(Generator polynomial)}
There is a unique sequence of monic $R$-polynomials
$g_0(X),g_1(X),\cdots,g_{\ell-1}(X)$ satisfying the following two:

\quad{\bf(g1)}
 $g_{0}(X)\,\big|\,(X^n-\lambda)$ and
$g_{v}(X)\,\big|\,g_{v-1}(X)$ for $v=1,\cdots,\ell-1$;

\quad{\bf(g2)}
$
 C=\big\langle g_{0}(X)+\pi g_{1}(X)
 +\cdots+\pi^{\ell-1}g_{\ell-1}(X)  \big\rangle.
$

\smallskip
{\bf(2)}{\rm(Check polynomial)}
There is a unique sequence of monic $R$-polynomials
$h_0(X),h_1(X),\cdots,h_{\ell-1}(X)$ satisfying the following two:

\quad{\bf(h1)}
 $h_{0}(X)\,\big|\,(X^n-\lambda)$ and
$h_{v}(X)\,\big|\,h_{v-1}(X)$ for $v=1,\cdots,\ell-1$;

\quad{\bf(h2)}
$h_{0}(X)+\pi h_{1}(X) +\cdots+\pi^{\ell-1}h_{\ell-1}(X)$
is a check polynomial of $C$.
\end{theorem}

\begin{proof} 
The following proof also shows how to get the generator polynomial 
and the check polynomial exhibited in the theorem from the ideal 
$$
C=\R \pi^{e_1}\wh\varphi_1(X)\oplus\cdots\oplus \R \pi^{e_m}\wh\varphi_m(X)
$$ 
of $\R =R[X]\big/\langle X^n-\lambda\big\rangle$,
see Eq.\eqref{generate C} and Eq.\eqref{check C}.

(1).~ By Theorem \ref{generator of ideal}, the integer sequence $e_1,\cdots,e_m$
determine a unique subset sequence 
$I_0\subseteq I_1\subseteq\cdots\subseteq I_{\ell-1}\subseteq I$, 
and the sequence of polynomials  
\begin{equation}\label{generate C}
 \wh\varphi_{I_0}(X),\wh\varphi_{I_1}(X),\cdots,\wh\varphi_{I_{\ell-1}}(X)
\end{equation}
satisfy the conditions (g1) and (g2).

For the uniqueness, we assume that 
$g_{0}(X),g_{1}(X),\cdots,g_{\ell-1}(X)$ satisfy (g1) and (g2),
and we'll show that $g_v(X)=\wh\varphi_{I_v}(X)$ for $v=0,1,\cdots,\ell-1$.
By Lemma~\ref{facts on local}(3), 
from (g1) we see that any $g_v(X)$ determines a unique subset
$I'_v\subseteq I$ such that $g_v(X)=\wh\varphi_{I'_v}(X)$, $v=0,1,\cdots,\ell-1$.
And by (g1) and Eq.\eqref{phi division},
 $I'_0\subseteq I'_1\subseteq\cdots\subseteq I'_{\ell-1}\subseteq I$,
which determines a unique integer sequence
$e'_1,\cdots,e'_m$, $0\le e'_i\le\ell$, hence 
\begin{align*}
C&=\R \pi^{e_1}\wh\varphi_1(X)\oplus\cdots\oplus\R \pi^{e_m}\wh\varphi_m(X)\\
&=\big\langle g_0(X)+\pi g_1(X)+\cdots+\pi^{\ell-1}g_{\ell-1}(X)\big\rangle\\
&=\big\langle \wh\varphi_{I'_0}(X)+\pi\wh\varphi_{I'_1}(X)
 +\cdots+\pi^{\ell-1}\wh\varphi_{I'_{\ell-1}}(X)\big\rangle\\
 &=\R \pi^{e'_1}\wh\varphi_1(X)\oplus\cdots\oplus\R \pi^{e'_m}\wh\varphi_m(X).
\end{align*}
where the last equality follows from Theorem \ref{generator of ideal}.
So, $e'_i=e_i$ for $i\in I$, hence $I'_v=I_v$ for $v=0,1,\cdots,\ell-1$.
That is, $g_v(X)=\wh\varphi_{I_v}(X)$ for $v=0,1,\cdots,\ell-1$.

\medskip (2).~
Obviously, for $i\in I=\{1,2,\cdots,m\}$,
$$
 {\rm Ann}_{\R \wh\varphi_i(X)}\big(\R \pi^{e_i}\wh\varphi_i(X)\big)
 =\R \pi^{\wh e_i}\wh\varphi_i(X), \qquad \wh e_i=\ell-e_i;
$$
where ${\rm Ann}_S(T):=\{s\in S\,|\,sT=0\}$ 
for any subset $T$ of any ring $S$. Thus
$$
{\rm Ann}_{\R }\big(C\big)=
\R \pi^{\wh e_1}\wh\varphi_1(X)\oplus\cdots\oplus\R \pi^{\wh e_m}\wh\varphi_m(X).
$$
The integer sequence $\wh e_1,\cdots,\wh e_m$ determine the subsets
$\wh I_v=\{i\mid i\in I, \wh e_i\le v\}$, $v=0,1,\cdots,\ell-1$,
such that
$\wh I_0\subseteq \wh I_1\subseteq\cdots \wh I_{\ell-1}\subseteq I$. 
Since
$$
\wh e_i=\ell-e_i\le v ~\iff~ e_i\ge \ell-v 
~\iff~ e_i\notin I_{\ell-v-1} ~\iff~ i\in I-I_{\ell-v-1},
$$
we have that 
$$\wh I_v=I-I_{\ell-v-1},  \quad v=0,1,\cdots,\ell-1;$$
By the definition in Eq.\eqref{phi_I}, 
$\wh\varphi_{I-I_{\ell-v-1}}(X)=\varphi_{I_{\ell-v-1}}(X)$. Thus
$$
\wh\varphi_{\wh I_v}(X)=\wh\varphi_{I-I_{\ell-v-1}}(X)
=\varphi_{I_{\ell-v-1}}(X),\quad v=0,1,\cdots,\ell-1.
$$
By Theorem \ref{generator of ideal}, 
\begin{align*}
{\rm Ann}_{\R }\big(C\big)
&=\big\langle  \wh\varphi_{\wh I_0}(X)+\pi\wh\varphi_{\wh I_1}(X)
 +\cdots+\pi^{\ell-1}\wh\varphi_{\wh I_{\ell-1}}(X)\big\rangle\\
&=\big\langle  \varphi_{I_{\ell-1}}(X)+\pi\varphi_{I_{\ell-2}}(X)
 +\cdots+\pi^{\ell-1}\varphi_{I_{0}}(X)\big\rangle.
\end{align*}
An $R$-polynomial $h(X)$ is a check polynomial of $C$
if and only if $h(X)$ is a generator of ${\rm Ann}_{\R }(C)$. 
By the above (1),
\begin{equation}\label{check C}
\varphi_{I_{\ell-1}}(X), \varphi_{I_{\ell-2}}(X), \cdots, \varphi_{I_{0}}(X),
\end{equation}
is the unique sequence of monic $R$-polynomials 
satisfying (h1) and (h2).
\end{proof}

Theorem~\ref{constacyclic cor2}(2) means that 
all $\lambda$-constacyclic codes are {\em checkable},
the essential reason is that $\R$ is principal
(of course, under the assumption at the beginning of this section).    

\subsection{A BCH bound for constacyclic codes over $R$}
From Theorem \ref{generator of ideal}(2), we have
more observations on $C\cap \R \pi^v$.

\begin{lemma}\label{R iso R'}
Let $0\le v\le\ell-1$, let $R'=R/R\pi^{\ell-v}$ which is a finite chain ring 
with $J(R')=R'\pi'$ of nilpotency index $\ell-v$
($\pi'$ denotes the image in $R'$ of $\pi$). Then
we have an $\R$-isomorphism ($\lambda'$ is the image in $R'$ of $\lambda$):
\begin{equation}\label{R R'}
\R \pi^{v}\cong R'[X]/\langle X^n\!-\!\lambda'\rangle,~~
\pi^v a(X)\,\longmapsto\, a'(X):=a(X)~({\rm mod}~\pi^{\ell-v}),
\end{equation}
which preserves Hamming weight.
\end{lemma}
\begin{proof}
The natural homomorphism 
$R\to R'$, $a\mapsto a'$, where $a'=a\!\pmod{\pi^{\ell-v}}$, 
induces a natural surjective homomorphism: 
$$\textstyle
 \R =R[X]/\langle X^n\!-\!\lambda\rangle
 \to R'[X]/\langle X^n\!-\!\lambda'\rangle,~~ 
 \sum\limits_{i=0}^{n-1}a_iX^i\mapsto \sum\limits_{i=0}^{n-1}a_i'X^i.$$
The kernel of the above homomorphism is 
$\R \pi^{\ell-v}$. We get an isomorphism:
\begin{equation}\label{iso R/pi }
\textstyle
\R /\R \pi^{\ell-v}\cong R'[X]/\langle X^n\!-\!\lambda'\rangle,~
\sum\limits_{i=0}^{n-1}a_iX^i~({\rm mod}~\pi^{\ell-v})\mapsto 
\sum\limits_{i=0}^{n-1}a_i'X^i.
\end{equation}
Next, consider the following surjective $\R $-homomorphism
$$
 \rho_{v}:~
 \R \longrightarrow \R \pi^{v}, 
  ~~ a(X)\longmapsto \pi^{v} a(X).
$$
Obviously, the kernel ${\rm ker}(\rho_{v})=\R \pi^{\ell-v}$, 
hence $\rho_v$ induces an $\R $-isomorphism
$$
\ol\rho_{v}:~ \R \big/\R \pi^{\ell-v} 
 \cong \R \pi^{v},~~ 
 a(X)~({\rm mod}~\pi^{\ell-v})\,\longmapsto\,\pi^v a(X).
$$
Combining it with Eq.\eqref{iso R/pi },
we get the isomorphism \eqref{R R'}.

Let $\pi^v a(X)\in\R\pi^v$ where $a(X)=\sum_{i=0}^{n-1}a_iX^i\in\R$.
Since $\pi^\ell=0$ but $\pi^{\ell-1}\ne 0$, we see that
 $\pi^v a_i=0$ if and only if $a_i\equiv 0~({\rm mod}~\pi^{\ell-v})$.
Thus, the Hamming weight 
$${\rm w_H}\big(\pi^v a(X)\big)={\rm w_H}\big(a'(X)\big),$$
where $a'(X)=\sum_{i=0}^{n-1}a_i'X^i
 \in R'[X]/\langle X^n\!-\!\lambda'\rangle$.
\end{proof}

\begin{lemma}\label{C cap R pi^v}
Let notation be as in Theorem \ref{generator of ideal}
and Lemma \ref{R iso R'}. Then:

{\bf(1)} Through the isomorphism \eqref{R R'},
$C\cap\R \pi^v$ is a 
$\lambda'$-constacyclic code over $R'$ of length $n$ 
with a generator polynomial 
$$\wh\varphi_{I_v}'(X)+\pi\wh\varphi_{I_{v+1}}'(X) 
 +\cdots+\pi^{\ell-v-1}\wh\varphi_{I_{\ell-1}}'(X)\in R'[X],$$  
where $\wh\varphi_{I_j}'(X)$ denotes the image in $R'[X]$ of 
$\wh\varphi_{I_j}(X)$. 

{\bf(2)} The minimum Hamming weight 
${\rm w_H}(C\cap\R \pi^v)={\rm w_H}(C)$.
\end{lemma}

\begin{proof}
(1). By Theorem \ref{generator of ideal}(2),  
$$C\cap\R \pi^{v}=
\big\langle\pi^v\big(\wh\varphi_{I_v}(X)+\pi\wh\varphi_{I_{v+1}}(X) 
 +\cdots+\pi^{\ell-v-1}\wh\varphi_{I_{\ell-1}}(X)\big)\big\rangle.$$
Through the isomorphism \eqref{R R'}, $C\cap \R \pi^v$
is an ideal of $R'[X]/\langle X^n-\lambda'\rangle$
with the generator polynomial stated in (1).

(2). Since $C\cap \R \pi^v\subseteq C$,
${\rm w_H}(C)\le {\rm w_H}(C\cap \R \pi^v)$.
For the other hand, let $c(X)=c_0+c_1X+\cdots+c_{\ell-1}X^{\ell-1}\in C$ 
be any non-zero code word. 
There is an integer $k\ge 0$ such that $0\ne\pi^kc(X)\in\R \pi^v$.
Thus $0\ne \pi^kc(X)\in C\cap\R \pi^v$,
hence ${\rm w_H}(\pi^kc(X))\ge{\rm w_H}(C\cap\R \pi^v)$.
If $\pi^kc_j\ne 0$ then $c_j\ne 0$ obviously.
That is, ${\rm w_H}(c(X))\ge {\rm w_H}(\pi^kc(X))$.
We obtain that ${\rm w_H}(c(X))\ge {\rm w_H}(C\cap\R \pi^v)$.
In conclusion, ${\rm w_H}(C)\ge{\rm w_H}(C\cap\R \pi^v)$.
\end{proof}

\begin{remark}\label{q-cosets}\rm
Let us recall how to get the decomposition 
of $X^n-\ol\lambda$ in \eqref{X^n-lambda in F}; 
see \cite{CDFL} for details please. 
Assume that $t={\rm ord}_{F^\times}(\ol\lambda)$ 
is the order of $\ol\lambda$ in the multiplicative unit group $F^\times$.
In a large enough extension of $F$ we have a primary $tn$-th root $\xi$ of unity 
such that $\xi^n=\ol\lambda$. 
As before, $\Z_{tn}$ denotes the integer residue ring modulo $tn$.
The roots of $X^n-\ol\lambda$ are 
one to one corresponding to the points of the subset
$$1+t\Z_{tn}:=\{1+tk\mid k\in\Z_{tn}\}
 =\{1,1+t,\cdots,1+(n-1)t\}\subseteq\Z_{tn},
$$
 i.e.,
$$\textstyle
X^n-\ol\lambda=\prod\limits_{j\in(1+t\Z_{tn})}(X-\xi^j) ~~~~
 \mbox{\big(in $\tilde F[X]$, where $\tilde F=F(\xi)$\big).}
$$
Recall that $q=p^r=|F|$. Since $t\mid q-1$, $q\in 1+t\Z_{tn}$ and
$1+t\Z_{tn}$ is partitioned in to $q$-cosets $Q_1,\cdots,Q_m$. 
Then in $F[X]$ we have the irreducible decomposition:
$$\textstyle
X^n-\ol\lambda=\ol\varphi_1(X)\cdots\ol\varphi_m(X),\quad
 \mbox{where }~ \ol\varphi_i(X)=\prod\limits_{j\in Q_i}(X-\xi^j).
$$
For any monic $R$-polynomial $g(X)\mid X^n-\lambda$,
by Lemma \ref{facts on local}(3) 
there is a unique $q$-invariant subset $Q\subseteq(1+t\Z_{tn})$
(i.e., $Q$ is a union of some of $Q_1,\cdots,Q_m$) such that
$\ol g(X)=\prod_{j\in Q}(X-\xi^j)$.
\end{remark}

\begin{definition}\rm
Let $C\subseteq \R $ be an ideal, and 
$g_0(X)+\pi g_1(X)+\cdots+\pi^{\ell-1} g_{\ell-1}(X)$
be the generator of $C$ satisfying (g1) and (g2) of Theorem \ref{constacyclic cor2}.
Let $Q\subseteq(1+t\Z_{tn})$ be the 
corresponding $q$-invariant subset of $g_{\ell-1}(X)$.
We call $Q$ the {\em residue zero set} of 
the $\lambda$-constacyclic code $C$ over $R$.

For any ideal $C$ described in Eq.\eqref{ideals of R}, 
by Theorem \ref{generator of ideal}(3) and its notation,  
$g_{\ell-1}(X)=\wh\varphi_{I_{\ell-1}}(X)=\varphi_{I-I_{\ell-1}}(X)$, 
where $I_{\ell-1}=\{i\,|\, 1\le i\le m,\;e_i\le\ell-1\}$, hence
$I-I_{\ell-1}=\{i\,|\, 1\le i\le m,\;e_i=\ell\}$.
Thus, the residue zero set of the $\lambda$-constacyclic code $C$ 
is: $Q=\bigcup_{e_i=\ell} Q_i$. 
\end{definition}

\begin{theorem}[BCH Bound]\label{BCH bound}
Let $C$ be a non-zero $\lambda$-constacyclic code over~$R$ of length~$n$,
and $Q$ be the residue zero set of $C$. If
$Q$ contains $d$ consecutive points 
(consecutive in $1+t\Z_{tn}$, not in $\Z_{tn}$),
then ${\rm w_H}(C)\ge d+1$.
\end{theorem}

\begin{proof} 
First we assume that $\ell=1$, i.e. $R=F$ is a field.
Then the theorem is just the known BCH bound for constacyclic codes
over finite fields (e.g., \cite{KS}). For convenience we sketch a proof.
For any $j\in Q$, $\xi^j$  is a root of any code word 
$c(X)=c_0+c_1X+\cdots+c_{n-1}X^{n-1}\in C$, i.e., 
$c_0\xi^{0j}+c_1\xi^j+\cdots+c_{n-1}\xi^{(n-1)j}=0$.
Suppose that $Q$ contains $1+tk, 1+t(k+1),\cdots, 1+t(k+d-1)$.
If $0\ne c_0+c_1X+\cdots+c_{n-1}X^{n-1}\in C$, then
$$
\begin{pmatrix}
1 &\xi^{1+kt}&\cdots&\xi^{(1+kt)(n-1)}\\
1 &\xi^{1+(k+1)t}&\cdots&\xi^{(1+(k+1)t)(n-1)}\\
\cdots &\cdots&\cdots&\cdots\\
1 &\xi^{1+(k+d-1)t}&\cdots&\xi^{(1+(k+d-1)t)(n-1)}
\end{pmatrix}
\begin{pmatrix}c_0\\ c_1\\ \vdots\\ c_{n-1} \end{pmatrix}=0
$$
For $0\le j_1<\cdots<j_d\le n-1$, the determinant of 
the $d\times d$ submatrix consisting of the $j_i$-th columns,  $i=1,\cdots,d$,
of the coefficient matrix is
$$
\det
\begin{pmatrix}
\xi^{(1+kt)j_1}&\cdots&\xi^{(1+kt)j_d}\\
\xi^{(1+(k+1)t)j_1}&\cdots&\xi^{(1+(k+1)t)j_d}\\
\cdots &\cdots&\cdots\\
\xi^{(1+(k+d-1)t)j_1}&\cdots&\xi^{(1+(k+d-1)t)j_d}
\end{pmatrix}
$$
The entries in the $i$'th column have $\xi^{(1+kt)j_{i}}$ in common.
So the determinant is
$$
\xi^{(1+kt)j_1}\cdots \xi^{(1+kt)j_d}\cdot
\det
\begin{pmatrix}
1 &\cdots& 1\\
\xi^{tj_1}&\cdots&\xi^{tj_d}\\
\cdots &\cdots&\cdots\\
\xi^{tj_1(d-1)}&\cdots&\xi^{tj_d(d-1)}
\end{pmatrix}\ne 0.
$$
Therefore, in the non-zero code word $(c_0,c_1,\cdots,c_{n-1})$,
there are at least $d+1$ non-zero entries.  

Next, assume that $\ell\ge 1$. 
In Lemma~\ref{C cap R pi^v}(2), taking ${v=\ell-1}$, 
we get that ${\rm w_H}(C)={\rm w_H}(C\cap \R\pi^{\ell-1})$.
And, $R/R\pi^{\ell-(\ell-1)}=R/\pi=F$ is a field.
By Lemma~\ref{C cap R pi^v}(1),  
$C\cap \R\pi^{\ell-1}$ can be regarded as an ideal
of $F[X]/\langle X^n-\overline\lambda\rangle$
with generator polynomial $\overline g_{\ell-1}(X)\in F[X]$ 
which has zero set~$Q$. By the conclusion for finite fields proved above, 
${\rm w_H}(C\cap\R\pi^{\ell-1})\ge d+1$.
We get the desired result ${\rm w_H}(C)\ge d+1$.
\end{proof}

\subsection{An example}
We conclude this section by an example to illustrate how to get 
the decomposition of $X^n-\lambda$, the generator polynomial, 
the check polynomial and the BCH bound 
of a $\lambda$-constacyclic code $C$ of length $n$ over $R$.

\begin{example}\rm
Take $R=\Z_{5^2}$, i.e., $\pi=p=5$, $\ell=2$ and $r=1$; hence $F=\ol R=\Z_5$, $q=5$.
Take $n=6$, $\lambda=4$. Then $\ol\lambda=-1$, 
$t={\rm ord}_{\Z_5^\times}(-1)=2$, $tn=12$. 
In $\Z_{12}$,
$$
 1+t\Z_{tn}=\{1,3,5,7,9,11\}
 =Q_1\cup Q_2\cup Q_3\cup Q_4, 
$$
$$
 Q_1=\{1,5\},~~ Q_2=\{3\},~~ Q_3=\{7,11\}.~~ Q_4=\{9\}.
$$
In $F[X]$,
\begin{align*}
X^6+1&=(X^3-3)\cdot(X^3-2)=(X^3-2^3)\cdot(X^3-3^3)\\
&=(X-2)(X^2+2X+4)\cdot(X-3)(X^2+3X+4). 
\end{align*}
Then $X^2+3X+4=X^2-2X-1$ is irreducible over $F$, 
and ${\rm GF}(5^2)=F(\xi)$, where $\xi^2-2\xi-1=0$.
Then $\xi^3=2$ and
\begin{align*}
\ol\varphi_1(X)=X^2-2X-1,~~ \ol\varphi_2(X)=X-2,\\ 
\ol\varphi_3(X)=X^2+2X-1,~~ \ol\varphi_4(X)=X-3.   
\end{align*}
In $\Z_{5^2}[X]$, 
\begin{align*}
X^6-\lambda&=X^6-4=(X^3+2)(X^3-2)=(X^3+3^3)(X^3-3^3)\\
&=(X+3)(X^2-3X+9)(X-3)(X^2+3X+9).
\end{align*}
We see that
\begin{align*}
\varphi_1(X)=X^2+3X+9,~~ \varphi_2(X)=X+3,\\
\varphi_3(X)=X^2-3X+9,~~ \varphi_4(X)=X-3.
\end{align*}
Thus
$$\R =\Z_{5^2}[X]\big/\langle X^6-4\rangle
=\R \wh\varphi_1(X)\oplus\R \wh\varphi_2(X)
\oplus\R \wh\varphi_3(X)\oplus\R \wh\varphi_4(X).
$$
The four ideals are all Galois rings: 
$$\begin{array}{c}
\R \wh\varphi_1(X)\cong \R \wh\varphi_3(X)\cong{\rm GR}(5^2,2),\\[3pt]
\R \wh\varphi_2(X)\cong \R \wh\varphi_4(X)\cong{\rm GR}(5^2,1)=\Z_{5^2}.
\end{array}$$
Take
$$
C=
\R \cdot 5\wh\varphi_3(X)\oplus\R \wh\varphi_4(X), 
$$
i.e., $e_1=2$, $e_2=2$, $e_3=1$, $e_4=0$;
$I_0=\{4\}$, $I_1=\{3,4\}$. By Theorem \ref{constacyclic cor2},
Eq.\eqref{generate C} and \eqref{check C},
$$\begin{array}{c}
g_0(X)=\varphi_1(X)\varphi_2(X)\varphi_3(X),~~ 
g_1(X)=\varphi_1(X)\varphi_2(X);\\[3pt]
h_0(X)=\varphi_3(X)\varphi_4(X),~~ h_1(X)=\varphi_4(X).
\end{array}$$
The $q$-invariant subset of $\ol g_1(X)=\ol\varphi_1(X)\ol\varphi_2(X)$
is $Q_1\cup Q_2=\{1,3,5\}$, which are consecutive in $1+2\Z_{12}$.
By Theorem \ref{BCH bound}, the minimal Hamming weight ${\rm w_H}(C)\ge 4$.
In fact, ${\rm w_H}(C)=4$, because:
$$
g_1(X)=(X^2+3X+9)(X+3)=2+18X+6X^2+X^3,
$$
hence, in $C\cap\R \pi=\R \pi g_1(X)$, we have a code word
$$5\cdot(2,18,6,1,0,0)=(10,15,5,5,0,0),$$
whose Hamming weight is equal to $4$.
\end{example}

\section{Constacyclic codes over finite PIRs}

We abbreviate ``principal ideal ring'' by ``PIR''.
In this section we investigate the question:  
when are the constacyclic codes over a finite PIR principal?

\begin{lemma}\label{R=R_1+...}
Let $R=R_1\oplus\cdots\oplus R_m$ be a finite ring, where
$R_1,\cdots,R_m$ are finite local rings. Let 
$\lambda=\lambda_1+\cdots+\lambda_m\in R^\times$,
where $\lambda_i\in R_i^\times$ for $i=1,\cdots,m$.
Let $n>1$ be an integer. Then
$$
R[X]\big/\langle X^n-\lambda\rangle
=R_1[X]\big/\langle X^n-\lambda_1\rangle
\oplus\cdots\oplus
R_m[X]\big/\langle X^n-\lambda_m\rangle.
$$
In particular, $R[X]\big/\langle X^n-\lambda\rangle$ is principal
if and only if every $R_i[X]\big/\langle X^n-\lambda_i\rangle$ 
for $1\le i\le m$ is principal.
\end{lemma}

\begin{proof}
Let $1=1_R=\iota_1+\cdots+\iota_m$ 
with $\iota_i\in R_i$ for $i=1,\cdots,m$;
i.e., $\iota_i$ is the identity of the local ring $R_i$.
For any $a\in R$, $a=\iota_1a+\cdots+\iota_m a$ and $\iota_ia\in R_i$.
Then, for any $a(X)=\sum_{j}a_jX^j\in R[X]$,
$$\textstyle
 a(X)=\iota_1 a(X)+\cdots+\iota_m a(X), \quad 
  \iota_i a(X)=\sum_{j}\iota_ia_jX^j\in R_i[X].
$$
Since $\iota_i R=R_i$, 
$$
 R[X]=\iota_1 R[X]\oplus\cdots\oplus \iota_m R[X]
 =R_1[X]\oplus\cdots\oplus R_m[X].
$$
And, $\iota_i\lambda=\lambda_i$ for $i=1,\cdots,m$, so
$$
X^n\!-\!\lambda =\iota_1(X^n\!-\!\lambda)+\cdots+\iota_m(X^n\!-\!\lambda)
=(\iota_1 X^n\!-\!\lambda_1)+\cdots+(\iota_m X^n\!-\!\lambda_m).
$$
Thus,
$$
R[X]\big/\langle X^n\!-\!\lambda\rangle
=R_1[X]\big/\langle \iota_1 X^n\!-\!\lambda_1\rangle
\oplus\cdots\oplus
R_m[X]\big/\langle \iota_m X^n\!-\!\lambda_m\rangle.
$$
Note that $\iota_i$ is just the identity $1_{R_i}$ of $R_i$.
As a conventional expression, 
$R_i[X]\big/\langle \iota_i X^n-\lambda_i\rangle=
R_i[X]\big/\langle  X^n-\lambda_i\rangle$.
\end{proof}

\begin{theorem}
Let $R$ be a finite PIR and $\lambda\in R^\times$. 
If $\gcd(n,{\rm char} R)=1$,
then the quotient ring $R[X]\big/\langle X^n-\lambda\rangle$ 
is a finite PIR.
\end{theorem}

\begin{proof}
By Eq.\eqref{being PIR}, $R=R_1\oplus\cdots\oplus R_m$, 
where $R_i$, $i=1,\cdots,m$, are finite chain rings. 
By Eq.\eqref{+local}, 
${\rm char} R\!=\!{\rm lcm}\big({\rm char} R_1,\cdots,{\rm char} R_m\big)$.
Since $\gcd(n,{\rm char} R)\!=\!1$, we get  
$\gcd(n,{\rm char} R_i)=1$ for $i=1,\cdots,m$.
By Lemma~\ref{R=R_1+...} and Theorem~\ref{constacyclic}, 
we obtain the desired result.
\end{proof}

However, the condition ``$\gcd(n,{\rm char} R)=1$'' is not necessary. 
The following result is known, e.g., see \cite{FZ}.
For convenience, we provide a shorter proof.

\begin{lemma}\label{over field} 
Let $F$ be a finite field, $\lambda\!\in\! F^\times$.
Then, for any $n>1$, 
$F[\!X\!]\big/\langle X^n\!-\!\lambda\rangle$ is a finite PIR.
\end{lemma}

\begin{proof} 
Let $|F|=p^r$. 
Let $n=n'p^e$ with $p\nmid n'$ 
and $e\ge 1$ (otherwise the lemma holds by Theorem \ref{constacyclic}). 
The map $a\mapsto a^{p^e}$ for $a\in F$ is an automorphism of $F$.
There is a $\lambda'\in F$ such that $\lambda=\lambda'^{p^e}$.
In $F[X]$, 
$
X^{n'}-\lambda'=\varphi_1(X)\cdots\varphi_m(X)
$
with $\varphi_1(X),\cdots,\varphi_m(X)$ being monic irreducible 
and coprime to each other. Then
$$
X^{n'p^e}-\lambda=X^{n'p^e}-\lambda'^{p^e}=
(X^{n'}-\lambda')^{p^e}
=\varphi_1(X)^{p^e}\cdots\varphi_m(X)^{p^e}
$$
with $\varphi_1(X)^{p^e}, \cdots,\varphi_m(X)^{p^e}$ coprime each other. 
Thus 
$$
F[X]\big/\big\langle X^n-1\big\rangle \cong 
F[X]\big/\big\langle \varphi_1(X)^{p^e} \big\rangle 
  \oplus \cdots \oplus 
    F[X]\big/\big\langle \varphi_m(X)^{p^e} \big\rangle.
$$
For $\R_i:=F[X]\big/\big\langle \varphi_i(X)^{p^e} \big\rangle$, 
$1\le i\le m$, $\varphi_i(X)$ is a nilpotent element and
$\R_i/\R_i\varphi_i(X)\cong F[X]/\langle\varphi_i(X)\rangle$
is a finite field. By Lemma \ref{being chain},
$\R _i$ is a finite chain ring.
\end{proof}

By Lemma \ref{R=R_1+...}, Lemma \ref{over field} and Theorem \ref{constacyclic}, 
we obtain a sufficient condition for the constacyclic codes being principal.

\begin{corollary}\label{a sufficiency}
Let $R=R_1\oplus\cdots\oplus R_m$ be a finite PIR, where
$R_i$ for $i=1,\cdots,m$, are finite chain rings 
with radical $J(R_i)$ of nilpotency index $\ell_i$. 
Let $\lambda\in R^\times$.
If $\min\{\ell_i,\gcd(n,{\rm char} R_i)\}=1$ for $i=1,\cdots,m$, 
then any constacyclic code over~$R$ of length~$n$ is principal.
\end{corollary}

The condition is in fact also necessary for cyclic codes.
We'll prove a little more general version. 

\begin{definition}\label{d iso cyclic}\rm
Let $R$ be a finite ring and $\lambda\in R^\times$.
If there is a ring isomorphism 
$\rho: R[X]\big/\langle X^n-\lambda\rangle
 \to R[X]\big/\langle X^n-1\rangle$, which preserves Hamming weights,
then $\rho$ is called an {\em isometry}, and
the $\lambda$-constacyclic codes over $R$ of length $n$
are called {\em isometrically cyclic} codes over $R$.
\end{definition}

\begin{lemma}\label{l iso cyclic}
Let $R$ be a finite ring, $\lambda\in R^\times$.
If $\gcd\big(n,{\rm ord}_{R^\times}(\lambda)\big)=1$
(where ${\rm ord}_{R^\times}(\lambda)$ denotes the order 
of $\lambda$ in the multiplicative unit group $R^\times$), 
then $\lambda$-constacyclic codes are isometrically cyclic.
\end{lemma}

\begin{proof}
Let $t={\rm ord}_{R^\times}(\lambda)$. So $\lambda^t=1$.
Since 
$\gcd(n,t)=1$, there are integers $\alpha,\beta$ such that $n\alpha+t\beta=1$.
Then $\lambda=\lambda^{n\alpha+t\beta}
=\lambda^{n\alpha}\lambda^{t\beta}=\lambda^{n\alpha}$.
The following is obviously a ring isomorphism.
$$ \rho_{}:~ R[X]\longrightarrow R[X],~~
  f(X)\longmapsto f(\lambda^{\alpha}X). 
$$
And 
$$\rho(X^n-\lambda)=(\lambda^{\alpha} X)^n-\lambda
=\lambda^{\alpha n} X^n-\lambda^{\alpha n}=\lambda^{\alpha n}(X^n-1).
$$
Since $\lambda^{\alpha n}$ is invertible in $R[X]$, 
$$
\rho\big(R[X](X^n-\lambda)\big)=R[X]\big(\lambda^{\alpha n}(X^n-1)\big)
=R[X](X^n-1).
$$
Thus the ring isomorphism $\rho$ induces a ring isomorphism
\begin{equation}\label{isometry}
\begin{array}{cccc}
\tilde\rho: & R[X]\big/\langle X^n-\lambda\rangle&\longrightarrow& 
  R[X]\big/\langle X^n-1\rangle,\\[5pt]
 & a(X)=\sum_{i=0}^{n-1}a_iX^i&\longmapsto& 
 a(\lambda^{\alpha}X)=\sum_{i=0}^{n-1}a_i\lambda^{\alpha i}X^i.
\end{array}
\end{equation}
Because $\lambda$ is invertible, $a_i\lambda^{\alpha i}=0$ if and only if
$a_i=0$. That is, 
$$
 {\rm w_H}(\tilde\rho(a(X)))={\rm w_H}(a(X)), \quad
   \forall~a(X)\in R[X]\big/\langle X^n-\lambda\rangle.
$$
In conclusion, Eq.\eqref{isometry} is an isometry.
\end{proof}

\begin{remark}\rm
It's a typical example that, as ${\rm ord}_{R^\times}\!(\!-\!1)\!=\!2$ 
(or $1$ once ${\rm char} R\!=\!2$), 
the negacyclic codes are isometrically cyclic provided the code length $n$ is odd.
\end{remark}

\begin{theorem}\label{cyclic is p.i.r.}
Let $R$ be a finite chain ring with radical
$J(R)=R\pi$ of nilpotency index $\ell$, let $\lambda\in R^\times$ such that
$\gcd\big(n,{\rm ord}_{R^\times}(\lambda)\big)=1$. 
Then $R[X]/\langle X^n-\lambda\rangle$ is a PIR
if and only if $\min\{\ell,\gcd(n,{\rm char} R)\}=1$. 
\end{theorem}

\begin{proof} The sufficiency is a special case of Corollary \ref{a sufficiency}. 

Let ${\rm char} R=p^s$. In the following we assume that $\ell\ge 2$,
$n=n'p^e$ with $e\ge 1$ and $p\nmid n'$, and prove that 
$R[X]/\langle X^n-\lambda\rangle$ is not a PIR. 
By Lemma \ref{l iso cyclic}, we can further assume that $\lambda=1$.
Let $F:=\ol R$ be the residue field.
In $F[X]$
we have an irreducible factorization: 
$$
X^{n'}-1=\ol\varphi_1(X)\ol\varphi_2(X)\cdots\ol\varphi_m(X),
$$
where $\ol\varphi_1(X)=X-1$. Then in $R[X]$ we have
$$
X^{n'}-1=\varphi_1(X)\varphi_2(X)\cdots\varphi_m(X),
$$
where $\varphi_i(X)$ is the Hensel's lifting of $\ol\varphi_i(X)$,
and $\varphi_1(X), \varphi_2(X), \cdots,\varphi_m(X)$ 
are basic irreducible and coprime to each other.
And, in the present case we have that $\varphi_1(X)=X-1$.
Thus,
$$X^{n}-1=(X^{p^e})^{n'}-1=\varphi_1(X^{p^e})
\varphi_2(X^{p^e})\cdots\varphi_m(X^{p^e}),
$$
and $\varphi_1(X^{p^e}), \varphi_2(X^{p^e}), \cdots,\varphi_m(X^{p^e})$ 
are still coprime to each other. By Chinese Remainder Theorem,
$$
R[X]\big/\big\langle X^n-1\big\rangle \cong 
R[X]\big/\big\langle \varphi_1(X^{p^e}) \big\rangle \oplus 
   R[X]\big/\big\langle \varphi_2(X^{p^e}) \big\rangle\oplus\cdots\oplus 
    R[X]\big/\big\langle \varphi_m(X^{p^e}) \big\rangle.
$$
It is enough to show that 
$$
\R :=R[X]\big/\big\langle \varphi_1(X^{p^e}) \big\rangle =
R[X]\big/\big\langle X^{p^e}-1 \big\rangle
$$ 
is a local ring but not a chain ring.

Note that $\R \pi$ is a nilpotent ideal of $\R $. And, in the quotient
$$\ol\R :=
 \R /\R \pi \cong F[X]\big/ F[X](X^{p^e}-1),
$$
the element $X-1\ne 0$ but $(X-1)^{p^e}=X^{p^e}-1=0$. 
Turn to $\R $, $(X-1)^{p^e}\in\R \pi$. Thus, $X-1$ 
is a nilpotent element of $\R $, which is not contained in $\R \pi$
(because $e\ge 1$).

Further, the ideal $\langle\pi,X-1 \rangle$ of $\R $ generated by $\pi$ and $X-1$
is a nilpotent ideal and 
(by First Isomorphism Theorem)
$$
\R \big/\langle\pi,X-1 \rangle=\ol\R \big/\langle X-1\rangle\cong 
F[X]\big/\langle X-1\rangle=F,
$$
which implies that $\R $ is a local ring
with $J(\R )=\langle\pi,X-1 \rangle$.

On the other hand, in $R[X]$, $X^{p^e}-1=(X-1)(X^{p^e-1}+\cdots+X+1)$,
so the ideal $R[X]\cdot(X^{p^e}-1)\subseteq R[X]\cdot(X-1)$, 
and we have an isomorphism
(by First Isomorphism Theorem again)
$$
 \R \big/\R (X-1)\cong R[x]\big/(R[X]\cdot(X-1))\cong R.
$$
Thus, $\pi$ is non-zero in $\R \big/\R (X-1)$. We conclude that
$\pi$ is a nilpotent element of $\R $ which is not contained in
the nilpotent ideal $\R (X-1)$.

Combining the above, we see that $\R $ is a local ring; 
and both $\R \pi$ and $\R (X-1)$ are nilpotent ideals, 
but non of them contains the other one. 
In conclusion, $\R $ is not a chain ring. 
\end{proof}

Combining Theorem \ref{cyclic is p.i.r.} with Lemma \ref{R=R_1+...},
we answer the question at the beginning of this section for the 
isometrically cyclic case immediately.

\begin{theorem}
Let $R=R_1\oplus\cdots\oplus R_m$ be a finite PIR, where
$R_i$ for $i=1,\cdots,m$, are finite chain rings 
with radical $J(R_i)$ of nilpotency index $\ell_i$. 
Let $\lambda\in R^\times$ such that 
$\gcd\big(n,{\rm ord}_{R^\times}(\lambda)\big)=1$. Then
$R[X]\big/\langle X^n-\lambda \rangle$ is a PIR if and only if
 $\min\{\ell_i,\gcd(n,{\rm char }R_i)\}=1$ for $i=1,\cdots,m$.
\end{theorem}

However, the necessity of  Theorem \ref{cyclic is p.i.r.} 
no longer holds for the case that 
$\gcd\big(n,{\rm ord}_{R^\times}(\lambda)\big)\ne 1$. 

\begin{example}\label{e galois}
Let $R={\rm GR}(p^r,s)$ to be the Galois ring of invariants $p,r,s$.
Take $\lambda=1+p\in R^\times$. Then $R[X]/\langle X^p-\lambda\rangle$ 
is finite chain ring of invariants $p,r,s,p,ps$.
\end{example}

\begin{proof} 
Denote $F=\ol R=R/Rp={\rm GF}(p^r)$, 
and denote $\R =R[X]/\langle X^p-\lambda\rangle$. 
Let $\pi=X-\lambda\in\R$. Then $\ol \lambda=1$ and
\begin{align*}
\ol\R:=\R/\R p=F[X]/\langle X^p-1\rangle,\quad
\ol\pi=X-1
\end{align*}
In $\ol\R$, $\ol\pi^p=X^p-1^p=0$. 
So $\pi\in\R p$, hence $\pi$ is a nilpotent element of $\R$. 
It is enough to prove that $\pi^p=pu$ for a unit $u\in\R^\times$. 
Because: from $p=u^{-1}\pi\in \R\pi$ we can get that 
\begin{align*}
 \R/\langle\pi\rangle
  &=\R/\langle p,\pi \rangle
   =R[X]/\langle p, X-\lambda, X^n-\lambda\rangle\\
  &=F[X]/\langle X^n-1, X-1 \rangle=F[X]/\langle X-1 \rangle=F;
\end{align*}
by Lemma \ref{being chain} and cf. Eq.\eqref{Eisenstein equation}, 
$\R$ is a finite chain ring of invariants $p,r,s,p, ps$. 

If $p=2$, then $\lambda=3$ and (note that $X^2=\lambda=3$)
$$ \pi^p= (X-3)^2=X^2-6X+9=3-6X+9=2(-3-3(X-3)); $$
since $3(X-3)=3\pi$ is a nilpotent element and $-3\in \R^\times$,
$u=-3-3(X-3)$ is a unit. We are done for the case that $p=2$.

In he following we assume that $p$ is odd.
For $0<j<p$, the binomial coefficient $\binom{p}{j}$ are divided by $p$; 
so we set $\binom{p}{i}/p=b_i$. Then
\begin{align}\label{pi^p=}
\pi^p=(X-\lambda)^p=\big((X-1)-p\big)^p
=(X-1)^p+p^2h(X),
\end{align}
where $h(X)=(-1)^p p^{p-2}+\sum_{i=1}^{p-1}(-1)^i b_i p^{i-1}(X-1)^{p-i}$.
Note that $X^p=1+p$ in~$\R$ and~$p$ is odd,
\begin{align}\label{(X-1)^p=}\textstyle
(X-1)^p=X^p+(-1)^p+p f(X)=p\big(1+f(X)\big), 
\end{align}
where $f(X)=\sum_{i=1}^{p-1}(-1)^i b_i X^{p-i}$. 
By Remainder Theorem, there is a $q(X)\in R[X]$ such that 
\begin{align}\label{f(X)=}
 f(X)=(X-\lambda)q(X)+f(\lambda)=\pi q(X)+f(\lambda).
\end{align} 
Since for $0<i<p$,
$(1+p)^{p-i}=1+pd_i$ where $d_i=\sum_{j=1}^{p-i}\binom{p-i}{j}p^{j-1}$,
we obtain
\begin{align*}\textstyle
f(\lambda)=\sum\limits_{i=1}^{p-1}(-1)^i b_i (1+p)^{p-i}
=pd+\sum\limits_{i=1}^{p-1}(-1)^ib_i
\end{align*}
where $d=\sum_{i=1}^{p-1}(-1)^i b_id_i$. 
Since $0=(1-1)^p=1-1+p\sum_{i=1}^{p-1}(-1)^ib_i$,
we see that $\sum_{i=1}^{p-1}(-1)^ib_i=0$.
By Eq.\eqref{f(X)=}, in $\R$ we get that $f(X)=\pi q(X)+pd$.
Combining it with Eq.\eqref{(X-1)^p=} and Eq.\eqref{pi^p=}, we obtain
\begin{align*}
\pi^p=p(1+\pi q(X)+pd) + p^2h(X)=p\big(1+\pi q(X)+pd+ph(X) \big).
\end{align*}
In $\R$, $\pi q(X)$, $pd$ and $ph(X)$ are all nilpotent, 
then so is the sum of them.
Thus $u=1+\pi q(X)+pd+ph(X)$ is a unit, and $\pi^p=pu$.
We are done for the example.
\end{proof}

\begin{remark}\rm
In Example \ref{e galois}, $R$ is a Galois ring, 
${\rm char} R=p^s$, $n=p$, $\lambda=1+p$. We take $s\ge 2$. 
Then $\ell=s>1$, $\gcd(n,{\rm char}R)=p>1$. 
But $R[X]/\langle X^p-\lambda\rangle$ is still a finite chain ring.
On the other hand,
\begin{itemize}
\item 
if $p$ is odd then ${\rm ord}_{R^\times}(\lambda)=p^{s-1}$;
\item
otherwise, $p=2$, 
 \begin{itemize}
  \item if $s=2$ then ${\rm ord}_{R^\times}(\lambda)=2$ 
   (it is just \cite[Example 6.4]{DL});
  \item if $s\ge 3$ 
   then ${\rm ord}_{R^\times}(\lambda)=2^{s-2}$.
 \end{itemize} 
\end{itemize}
Thus, $\gcd\big(n,{\rm ord}_{R^\times}(\lambda)\big)=p\ne 1$. 
The necessity of Theorem \ref{cyclic is p.i.r.} no longer holds. 

In conclusion, the question at the beginning of this section is answered for
isometrically cyclic codes over finite PIRs. 
However, it is still open for general case.
\end{remark}


\end{document}